\documentclass[a4paper, 11pt, abstract=true]{scrartcl}

\usepackage{amsmath}
\usepackage{amsthm}
\usepackage{amsfonts}
\usepackage{amssymb}
\usepackage{bbm}
\usepackage{graphicx}

\newtheorem{thm}{Theorem}
\newtheorem{lem}{Lemma}
\newtheorem{prop}[thm]{Proposition}
\newtheorem{cor}{Corollary}

\newtheorem{innercustomthm}{Theorem}
\newenvironment{customTheorem}[1]
  {\renewcommand\theinnercustomthm{#1}\innercustomthm}
  {\endinnercustomthm}

\newtheorem{dfn}{Definition}

\newcommand{\id}{\mathbbm{1}}

\newcommand{\diag}{\operatorname{diag}}

\usepackage{abstract}

\usepackage{hyperref}

\title{\vspace{-15mm}%
	Sinkhorn normal form for\\ unitary matrices}
\author{%
	\large
	\textsc{Martin Idel, Michael M. Wolf} \\[2mm]
	\normalsize	Zentrum Mathematik, Technische Universit\"{a}t M\"{u}nchen \\
	\vspace{-5mm}
	}
\date{}

\begin{document}

\maketitle

\begin{abstract}
Sinkhorn proved that every entry-wise positive matrix can be made doubly stochastic by multiplying with two diagonal matrices. In this note we prove a recently conjectured analogue for unitary matrices:  every unitary can be decomposed into two diagonal unitaries and one whose row- and column sums are equal to one. The proof is non-constructive and based on a reformulation in terms of symplectic topology. As a corollary, we obtain a decomposition of unitary matrices into an interlaced product of unitary diagonal matrices and discrete Fourier transformations. This provides a new decomposition of linear optics arrays into phase shifters and canonical multiports described by Fourier transformations. \end{abstract}
	
\section{Introduction}
For every $n\times n$ matrix $A$ with positive entries there exist two diagonal matrices $L,~R$ such that $LAR$ is doubly stochastic, i.e. the entries of each column and row sum up to one. This result was first obtained by Sinkhorn \cite{sin64}, who also gave an algorithm how to compute $L$ and $R$ by iterated left and right multiplication of diagonal matrices. 

Recently, De Vos and De Baerdemacker studied the same problem for unitary matrices \cite{vos14a}. They conjectured that for every $n\times n$ unitary $U$ there exist two unitary diagonal matrices $L, R$ such that $LUR$ has all row and column sums equal to one. To support their conjecture, they construct an algorithm similar to the iteration procedure for matrices with positive entries from \cite{sin64,sin67}. They also provide numerical evidence that the algorithm always converges to a unitary matrix with row and column sums equal to one. 

The goal of this paper is to prove the conjecture of De Vos and De Baerdemacker that such a normal form always exists by reformulating the problem in terms of symplectic topology. It turns out that the reformulated problem is a special case of the Arnold (sometimes Arnold-Givental) conjecture on the intersection of Lagrangian submanifolds \cite{mcd98}, which was solved for this case in \cite{bir04, cho04}. More precisely, in section \ref{sec:unit} we show:
\begin{customTheorem}{2} 
For every unitary matrix $U\in U(n)$ there exist two diagonal unitary matrices $L,R\in U(n)$ such that $A:=LUR$ satisfies $\sum_j A_{ji} =\sum_j A_{ij} =1$ for all $i=1,\ldots n$.
\end{customTheorem}
For a given unitary $U\in U(n)$ the triple $(L,R,A)$ is certainly not unique, since multiplying $L$ by a global phase and $R$ by its inverse does not change $LAR$. Hence, it makes sense to consider the decomposition $U=e^{i\varphi}L^{\prime}AR^{\prime}$, where $L^{\prime},R^{\prime}$ are unitary diagonal such that $L^{\prime}_{11}=R^{\prime}_{11}=1$ and $\varphi\in[0,2\pi)$. In particular, for $U(2)$, a simple complete solution was given in \cite{vos14a} from which one can see that for every non-diagonal matrix, there are only two different $A$ such that $e^{i\varphi}LAR=U$. For $n>2$ the picture is less clear and the reformulation in terms of symplectic topology appears to give further insight into the freedom of the decomposition.

In addition to the Sinkhorn-type normal form above, in section \ref{sec:derived} we give several reformulations that might be interesting for applications, for instance regarding the decomposition of general $2n-$port linear optics devices into canonical multiports and phase shifters.

\section{Sinkhorn-type normal form} \label{sec:unit}
In order to prove the decomposition theorem, we reformulate the problem of rescaling a unitary matrix into a problem in symplectic topology.  For the reader's convenience, necessary results including elementary calculations and definitions are included in \ref{sec:sympprel}. We only repeat the most important definitions for our reformulation. Recall that the complex projective space $\mathbb{C}P^n$ consists of all equivalence classes of $\mathbb{C}^{n+1}\backslash\{0\}$ w.r.t. $x\sim y\Leftrightarrow x=\lambda y$ with $\lambda\in\mathbb{C}\backslash\{0\}$.
\begin{dfn} \label{dfn:clifftor}
 The \emph{Clifford Torus} is the $n$-dimensional torus embedded in $\mathbb{C}P^n$, i.e. the set of points
\begin{align}
	T^n:=\{[w_0,\ldots, w_n]\in\mathbb{C}P^{n}\big||w_0|=|w_1|=\ldots=|w_n|\}.
\end{align}
\end{dfn}
This torus, as shown in the appendix in proposition \ref{prop:cliff}, is a Lagrangian submanifold of the symplectic manifold $\mathbb{C}P^n$. We obtain the following connection to our normal form:
\begin{lem} \label{lem:reform}
For any unitary $U\in U(n)$, there exist diagonal unitaries $L$ and $R$ such that $A:=LUR$ has row and column sums equal to one if and only if the Clifford torus $T^{n-1}\subset \mathbb{C}P^{n-1}$ fulfills $T^{n-1}\cap UT^{n-1}\neq \emptyset$. 
\end{lem}
\begin{proof}
Let $U\in U(n)$ be arbitrary but fixed. We first consider the usual torus $\mathbb{T}^{n}\subset \mathbb{C}^{n}$, i.e. the set of all vectors for which each component has modulus one:
\begin{align*}
	\mathbb{T}^n:=\{(e^{i\phi_1},\ldots,e^{i\phi_n})\subset \mathbb{C}^n\,|\,\phi_j\in\mathbb{R}\}
\end{align*}
Let us first show that the existence of a normal form is equivalent to $\mathbb{T}^n\cap U\mathbb{T}^n\neq \emptyset$. For one direction, let $\varphi\in\mathbb{T}^n$ such that $U\varphi\in\mathbb{T}^n$, i.e. $\varphi\in\mathbb{T}^n\cap U\mathbb{T}^n$. Define the two diagonal matrices $R^{-1}:=\diag(\varphi_1,\ldots,\varphi_n)\in U(n)$ and  $L^{-1}:=\diag((U\varphi)_i^{-1})=\diag((\overline{U\varphi})_i)\in U(n)$. With $A:=L^{-1}UR^{-1}$ and $e:=(1,\ldots,1)^{T}$ we obtain:
\begin{align*}
	Ae=L^{-1}U\varphi=e
\end{align*}
Likewise, since $\overline{A}e=Ae$ and $A$ is unitary, we obtain 
\begin{align*}
	A^Te&=A^T\overline{A}e=e.
\end{align*}
so that columns and rows of $A$ sum up to one. 

For the other direction, suppose $U=LAR$ is a decomposition as proposed. Then $\varphi:=R^{-1}e\in\mathbb{T}^n$ and
\begin{align*}
	U\varphi=LAR\varphi=LAe=Le\in\mathbb{T}^n
\end{align*}
hence $U\varphi\in\mathbb{T}^n\cap U\mathbb{T}^n$. 

The next step is to reformulate the problem using the Clifford torus. Clearly, $T^{n-1}\cap UT^{n-1}\neq \emptyset$ iff $(\lambda \mathbb{T}^n)\cap U\mathbb{T}^n\neq \emptyset$ for some $\lambda \in \mathbb{C}\setminus \{0\}$. Since $U$ is norm preserving, any intersection requires $|\lambda|=1$ so that 
\begin{align*}
	T^{n-1}\cap UT^{n-1}\neq \emptyset \quad \Leftrightarrow \quad \mathbb{T}^n\cap U\mathbb{T}^n\neq \emptyset.\end{align*}
\end{proof}

One of the main conjectures in symplectic topology, the Arnold or Arnold-Givental conjecture, states that a Lagrangian submanifold and its image under a Hamiltonian isotopy intersect at least as often as the sum of the $\mathbb{Z}_2$-Betti-numbers. For $T^n$, this sum is not zero, thus, using proposition \ref{prop:unitham}, Arnold's conjecture states in particular that $T^n$ should intersect with $UT^n$ at least once. While the Arnold conjecture is wrong in all generality and most cases are unknown, there is a positive result to the weaker question whether the torus intersects with its displaced version (c.f. \cite{bir04,cho04}). In order to formulate this result, we need the following:
\begin{dfn}
Let $(\mathcal{M},\omega)$ be a closed symplectic manifold with Hamiltonian symplectomorphisms $\mathrm{Ham}(\mathcal{M})$. A Lagrangian submanifold $\mathcal{L}\subset\mathcal{M}$ is called \emph{displaceable} by a Hamiltonian diffeomorphism, if there exists a $\psi\in\mathrm{Ham}(\mathcal{M})$ such that
\begin{align*}
	\mathcal{L}\cap\psi\mathcal{L}=\emptyset.
\end{align*}
\end{dfn}

The definition is slightly different from the one in \cite{bir04}, where the authors only consider nonempty open sets such that the restriction of $\omega$ to these sets is exact. However, they prove that the torus $T^n$ is displaceable in the above definition, if and only if there exists an open neighborhood $\mathcal{V}\supset T^n$ such that $\omega|_{\mathcal{V}}$ is exact and $\mathcal{V}$ is displaceable. With this we can state the final and crucial ingredient in the proof of the normal form:
\begin{thm}[\cite{bir04} theorem 1.3] \label{thm:thm13}
The Clifford torus $T^n\subset\mathbb{C}P^n$ cannot be displaced from itself by a Hamiltonian isotopy.
\end{thm}

Because every unitary matrix defines a Hamiltonian isotopy (see proposition \ref{prop:unitham} in the appendix), the theorem tells us in particular $T^n\cap UT^n\neq \emptyset$ for all unitaries $U\in U(n)$ so that together with lemma \ref{lem:reform} this proves the sought normal form:
\begin{thm} \label{thm:unitarynormal}
For every unitary matrix $U\in U(n)$ there exist two diagonal unitary matrices $L,R\in U(n)$ such that $A:=LUR$ fulfills $\sum_j A_{ji} =\sum_j A_{ij} =1$ for all $i=1,\ldots n$.
\end{thm}

\section{Equivalent normal forms for unitary matrices} \label{sec:derived}
To obtain equivalent normal forms, consider the $n\times n$ dimensional complex matrix $F_n$ with entries $(F_n)_{kl}:=\frac{1}{\sqrt{n}} \operatorname{exp}(\frac{2\pi i}{n} kl)$ with $k,l\in\{0,\ldots,n-1\}$, which is known as the \emph{discrete Fourier transformation}. It is easy to see that $F_n^{-1}=F^{\dagger}$, hence $F_n\in U(n)$. If we denote the standard basis of $\mathbb{C}^n$ by $\{e_i\}_{i=0}^{n-1}$ and $e:=(1,\ldots,1)^T$, then
\begin{align*}
	F_ne_0=F_n^{\dagger}e_0=\frac{e}{\sqrt{n}}.
\end{align*}
Now let $A\in U(n)$ be such that $Ae=A^Te=e$. Then
$
	F_n^{\dagger}AF_ne_0=e_0
$
and similarly, $(F_n^{\dagger}AF_n)^{T}e_0=F_nA^TF_n^{\dagger}e_0=e_0$, which shows that 
\begin{align*}
	F_n^{\dagger}AF_n=\begin{pmatrix}{} 1 & 0_{n-1}^T \\ 0_{n-1} & \tilde{U} \end{pmatrix}
\end{align*}
where $0_{n-1}:= 0\in\mathbb{C}^{n-1}$ and $\tilde{U}\in U(n-1)$. Thus, given a unitary $U\in U(n)$, we know that there exists a decomposition
\begin{align}
	U=LF_n\begin{pmatrix}{} 1 & 0_{n-1}^T \\ 0_{n-1} & \tilde{U} \end{pmatrix}F_n^{\dagger}R
\end{align}
with $\tilde{U}\in U(n-1)$ and diagonal $L,R\in U(n)$. We can now iterate the procedure by applying it to the $(n-1)\times(n-1)$-dimensional submatrix $\tilde{U}$ and obtain the corollary:
\begin{cor} \label{cor:optics}
Let $U\in U(n)$, then there exist diagonal unitaries $D_1,\ldots,D_n$ and $\tilde{D}_1,\ldots,\tilde{D}_{n-1}$ and a $\varphi\in[0,2\pi)$ such that the first $i-1$ entries in each $D_i,\tilde{D}_i$ are equal to one and
\begin{align}
\begin{split}
	U&=D_1F_nD_2(\id_1\oplus F_{n-1})D_3(\id_2\oplus F_{n-2})\cdots \\
	&~~ D_{n-1}(\id_{n-2}\oplus F_{2})D_n(\id_{n-2}\oplus F_2^{\dagger})\tilde{D}_{n-1}\cdots (\id_1\oplus F_{n-1}^{\dagger})\tilde{D}_2F_n^{\dagger}\tilde{D}_1 e^{i\varphi}.
\end{split}
\end{align}
\end{cor}
\begin{figure}[!t]
  \centering
   \includegraphics[width=0.8\columnwidth]{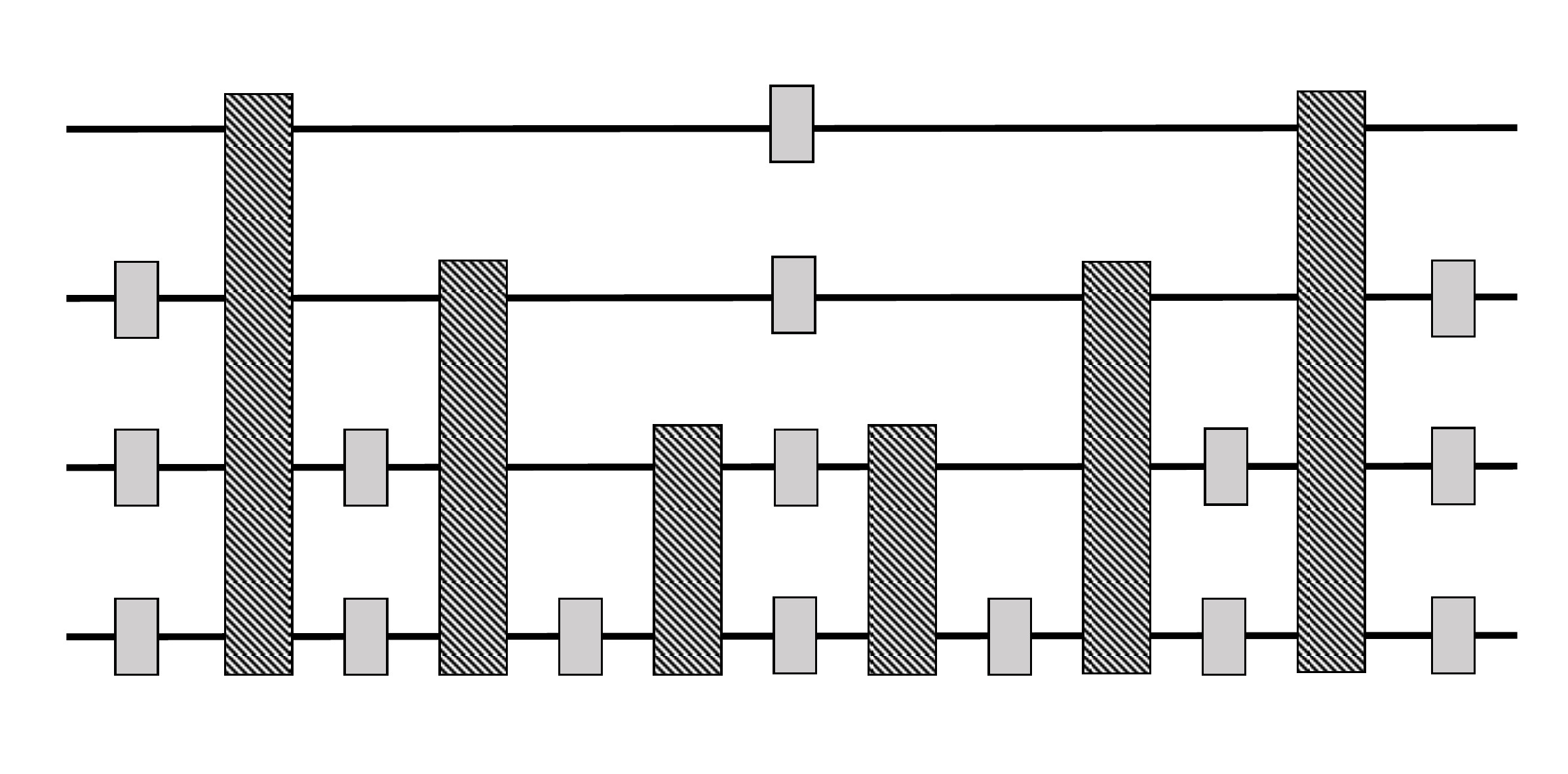}
  \caption{In quantum optics, passive transformations on $n$ modes are in one-to-one correspondence with $n\times n$ unitaries. Each unitary $U$ admits a decomposition into $2(n-1)$ canonical multiports (which are independent of $U$ and described by discrete Fourier transformations [hatched]) surrounded by $2n-1$ layers of single-mode phase shifters [grey]. Here, this is exemplified for $n=4$.}\label{multiport}
\end{figure}
In other words any unitary can be decomposed into diagonal unitaries and discrete Fourier transformations in this way. From the theorem, the first $k-1$ entries of the diagonal matrices $D_k,\tilde{D}_k$ are immediately known to be one. However one can achieve a better parameterisation by realizing that one can fix the first $k$ entries of $D_k,\tilde{D}_k$ to one for $k\leq n-1$, while absorbing all phases of the $k$-th entries of $D_k$ and $\tilde{D}_k$ into a diagonal unitary that replaces $D_n$ (this is immediately clear from a graphical representation as in Figure \ref{multiport}). 

This decomposition has an immediate application in quantum optics, where any $n\times n$ unitary corresponds to a passive transformation on $n$ modes or a $2n-$multiport. In this scenario a diagonal unitary corresponds to a set of phase shifters, which are applied to the modes individually and the discrete Fourier transformation is known as canonical $2n$-multiport \cite{mat95}, which may be implemented by a symmetric fibre coupler. The structure of the corresponding decomposition is graphically depicted in Figure \ref{multiport}.
   
Another version of the normal form is found by using that $D$ is a diagonal matrix iff $FDF^{\dagger}$ is a \emph{circulant matrix}, i.e. $(FDF^{\dagger})_{i,j}=:\alpha_{i-j}\in\mathbb{C}$. Since the diagonal matrices form a group, so do the circulant matrices and we denote the group of $n\times n$ circulant matrices by $\mathrm{Circ}(n)$. Then:
\begin{cor}
Let $U\in U(n)$, then there exist $C_1,C_2\in\mathrm{Circ}(n)$ and $\tilde{U}\in U(n-1)$ such that
\begin{align}
	U=C_1\diag(1,\tilde{U})C_2.
\end{align}
\end{cor}
Let us finally discuss the question of uniqueness of these decompositions and to this end come back to the original normal form
\begin{align}
	U=e^{i\varphi}D_1AD_2 \label{eqn:nform},
\end{align}
where $D_1,D_2$ are unitary diagonal with $(D_i)_{11}=1$ and $A$ has row and column sums equal to 1. Counting parameters, using that the matrices $A$ are isomorphic to $U(n-1)$ as proven above, we have:
\begin{align*}
	1+(n-1)+(n-1)^2+(n-1)=n^2
\end{align*}
parameters (c.f. \cite{vos14a}). Hence, the number of parameters matches exactly the dimension of $U(n)$. Given a unitary $U=e^{i\varphi}D_1AD_2$ as above, this means that it might be reasonable to expect only a discrete set of different decompositions or at least a discrete set of $A$ that $U$ can be scaled to. The exact number of different $A$ can easily be seen to be two for the case $n=2$ (c.f. \cite{vos14a}), but already for $n=3$ and $n=4$, there is only a conjectured bound (6 and 20, c.f. \cite{shc13}). 

In \cite{cho04} it is proven that if $T^n$ and $UT^n$ intersect transversally, their number of distinct intersection points must be at least $2^n$, which follows from general results in Floer-homology theory when applied to Lagrangian intersection theory. Since transversality is a generic property for intersections, one might therefore conjecture that for a generic unitary $U\in U(n)$ \cite{cho04} implies a lower bound $2^{n-1}$ on the number of different normal forms. However, it is not true that we always have a discrete number of decompositions or (in contrast to the $2\times 2$ case) at least a discrete number of $A$ such that $A$ has row and column-sums equal to one and $e^{i\varphi}LAR=U$. A counterexample is given by the Fourier transform in $4\times 4$ dimensions, where we have for any $\varphi\in[0,2\pi)$:\footnote{We thank the anonymous referee for providing this counterexample.}

\begin{align}
	\begin{split} \frac{1}{2}
		\begin{pmatrix}{} 1 &  1 &  1 &  1 \\ 
			1 &  i & -1 & -i \\
			1 & -1 &  1 & -1 \\
			1 & -i & -1 &  i
		\end{pmatrix}
		= \begin{pmatrix}{} 1 & 0 & 0 & 0 \\ 
			0 & e^{i\varphi} & 0 & 0 \\ 
			0 & 0 & 1 & 0 \\ 
			0 & 0 & 0 & -e^{-i\varphi} \\
		\end{pmatrix}\cdot \\ \frac{1}{2}
		\begin{pmatrix}{} 1 &  -ie^{i\varphi} &  1 &  ie^{i\varphi} \\ 
			e^{-i\varphi} &  1 & -e^{-i\varphi} & 1 \\
			1 & ie^{i\varphi} &  1 & -ie^{i\varphi} \\
			-e^{-i\varphi} & 1 & e^{i\varphi} &  1
		\end{pmatrix} \cdot
		\begin{pmatrix}{} 1 & 0 & 0 & 0 \\ 
			0 & ie^{-i\varphi} & 0 & 0 \\ 
			0 & 0 & 1 & 0 \\ 
			0 & 0 & 0 & -ie^{i\varphi}
		\end{pmatrix}
	\end{split}
\end{align}

After completion of this document, we learned that part of this section, in particular corollary \ref{cor:optics} were independently found in \cite{vos14b}.

\section{Conclusion}
We have studied a variant of a Sinkhorn type normal form for unitary matrices. Its existence was conjectured in \cite{vos14a} and we give a nonconstructive proof. This means in particular that the question, whether the algorithm presented in \cite{vos14a} always converges for any set of starting conditions, remains open. Also, it would be nice to have an elementary proof of the fact that for any unitary matrix $U$ we have $T^n\cap UT^n\neq \emptyset$. The decomposition is in not unique: We provided an example where, contrary to the $2\times 2$-case, there is a one-parameter set of $A$ as well as $L$ and $R$, such that $LAR=U$. We suggested an argument that the number of different decompositions, if it is discrete, might grow exponentially. However this lower bound relies on a lower bound on Lagrangian intersections which holds only for transversal intersections.

\subsection*{Acknowledgements}
We thank Michael Keyl for many helpful comments on the parts involving symplectic topology. M. Idel is supported by the Studienstiftung des deutschen Volkes. M. Wolf acknowledges support from the CHIST-ERA/BMBF project CQC. 

\bibliographystyle{alpha}
\bibliography{literatur}

\newpage

\appendix

\section{Symplectic Preliminaries} \label{sec:sympprel}
This section introduces the definitions and results from symplectic topology beyond the first chapters of \cite{mcd98} needed to understand the basic reductions of the proof of theorem \ref{thm:thm13} in \cite{bir04}.

\subsection{Notation and basic definitions}
To fix notation, a symplectic manifold will always be denoted by $\mathcal{M}$ and its symplectic form will be called $\omega$. The group of \emph{symplectomorphisms} of a symplectic manifold $(\mathcal{M},\omega)$ will be denoted by $\mathrm{Symp}(\mathcal{M})$ and its \emph{Hamiltonian symplectomorphisms} (i.e. all symplectomorphisms which are elements of the flow of a Hamiltonian vector field) will be denoted by $\mathrm{Ham}(\mathcal{M})$. We have the following characterization (\cite{mcd98}, chapter 10):
\begin{prop}
Let $(\mathcal{M},\omega)$ be a closed symplectic manifold. If the manifold is simply connected (i.e. every loop is contractible)
\begin{align*}
	\mathrm{Ham}(\mathcal{M})=\mathrm{Symp}_0(\mathcal{M})
\end{align*}
where $\mathrm{Symp}_0(\mathcal{M})$ denotes the connected component of the identity of the whole group of symplectomorphisms.
\end{prop}

In principle, the result also holds for arbitrary symplectic manifolds. One has to be more careful with non-compactly supported functions, but we can safely ignore these subtleties, since our manifold of interest will be closed.

Furthermore, let us recall that a \emph{Lagrangian submanifold} $\mathcal{L}$ of a $2n$-dimensional symplectic manifold $(\mathcal{M},\omega)$ is a smooth $n$-dimensional submanifold of $\mathcal{M}$ such that 
\begin{align*}
	T_p\mathcal{L}^{\varepsilon}:=\{X\in T_p\mathcal{M}|\omega(X,Y)=0~\forall\,Y\in T_p\mathcal{L}\}=T_p\mathcal{L} \quad \forall p\in\mathcal{L}
\end{align*}

\subsection{The Clifford-torus as a Lagrangian submanifold}
We now study the Clifford torus as a special case of the Lagrangian submanifold of interest for our result.

Before proving that the Clifford torus is a Lagrangian submanifold, we need to specify the symplectic structure on $\mathbb{C}P^n$: Consider the map $\Phi:\mathbb{C}^{n+1}\setminus\{0\}\to \mathbb{S}^{n+1}\subset \mathbb{C}^ {n+1}$ via $z\mapsto z/|z|$. We will show that the pullback $\Phi^*\omega$ of the standard symplectic structure $\omega$ on $\mathbb{C}^{n+1}$ descends to a symplectic form $\omega_{FB}$ on $\mathbb{C}P^n$, the \emph{standard symplectic structure} or \emph{Fubini-Study form} of the complex projective space. 
\begin{prop} \label{prop:cliff}
$\mathbb{C}P^n$, equipped with the Fubini-Study form is a $2n$-dimensional symplectic manifold and the Clifford Torus is a Lagrangian submanifold thereof.
\end{prop}
\begin{proof}
Let us go through the construction in more detail and see, how it defines a symplectic form, e.g. a non-degenerate and closed 2-form on $\mathbb{C}P^n$. Throughout, we will consider the natural projection $\pi: \mathbb{C}^{n+1}\setminus\{0\}\to \mathbb{C}P^n$. 

Note that if $(x_0,y_0,\ldots, x_n,y_n)$ are the real coordinates of $\mathbb{R}^{2n+2}\cong \mathbb{C}^{n+1}$, we can use $(z_0,\overline{z}_0,\ldots,z_n,\overline{z}_n)$ as coordinates for any point $(z_0,\ldots,z_n)\in\mathbb{C}^{n+1}$ as well. Then the standard symplectic form reads
\begin{align*}
	\omega=\sum_j dx^j\wedge dy^j=\frac{i}{2}\sum_j dz^j\wedge d\overline{z}^j
\end{align*}
Considering the action of $\mathbb{C}^*$ on $\mathbb{C}^{n+1}$, we obtain $\omega_{\lambda\cdot z}=\frac{i}{2}\sum_j d(\lambda\cdot z^j)\wedge d\overline{(\lambda\cdot z^j)}=|\lambda|^2\frac{i}{2}\sum_j dz^j\wedge d\overline{z}^j=|\lambda|^2\omega_z$. Hence, if $\Phi:\mathbb{C}^{n+1}\setminus \{0\}\to\mathbb{S}^{2n+1}$ is given by $z\mapsto z/|z|$, then $\Phi^*\omega$ will be invariant under the action of $\mathbb{C}^*$. This shows that $\Phi^*\omega$ descends to a well-defined 2-form $\omega_{FS}$ on $\mathbb{C}P^n$, by defining:
\begin{align*}
	(\omega_{FS})_{\pi(p)}(d\pi X_{\pi(p)}, d\pi Y_{\pi(p)})=(\Phi^*\omega)_p(X,Y)
\end{align*}

The next step is to show non-degeneracy. For this, note that $\Phi^*\omega (X,Y)=0~\forall Y$ if and only if $d\Phi X=0$ pointwise, since $\omega$ is non-degenerate. But $d\Phi X=0$ implies in particular $d\pi X=0$ and hence, $\omega_{FS}$ as defined above is a non-degenerate 2-form.

Finally, we need to prove closedness. This can either be computed directly by considering coordinates, or by considering local sections of the projection $\pi$. Let $\{U_i\}_i$ be a cover of $\mathbb{C}P^n$ such that there exist local section $\sigma_i:U_i\to \mathbb{C}^{n+1} \setminus \{0\}$. On each $U_i$ we have $\omega_{FS}=\sigma_i^*\Phi^*\omega$. But then
\begin{align*}
	d\omega_{FS}=d(\sigma_i^*\Phi^*\omega)=(\sigma_i\Phi)^*d\omega=0
\end{align*}
since $d$ commutes with pullbacks and $\omega$ is closed. Since this holds on any patch $U_i$, $d\omega_{FS}=0$ globally. 

In addition, we need to see that the Clifford torus is a Lagrangian submanifold. It is easy to see that the Clifford torus is a submanifold of (real) dimension $n$, hence we only need to prove $(T_pT^n)^\varepsilon=T_pT^n~\forall p\in T^n$. Given the canonical projection $\pi:\mathbb{C}^{n+1}\setminus \{0\}\to \mathbb{C}P^n$, $T^n$ is the image of $\pi$ of the torus
\begin{align*}
	\mathbb{T}^n:=\{(z_0,\ldots, z_n)||z_0|=|z_1|=\ldots=|z_n|=1\}
\end{align*}
By inspection, we obtain for $p=(p_0,\ldots,p_n)\in\mathbb{C}^{n+1}\setminus \{0\}$:
\begin{align*}
	T_p\mathbb{T}^n=\operatorname{span}\{p_i\partial_{\overline{p}_i}-\overline{p}_i\partial_{p_i}|i=0,\ldots,n\}=:\operatorname{span}\{X^i_p|i=0,\ldots,n\}
\end{align*}
Then $T_{\pi (p)} T^n$ will be spanned by $d\pi X^i_{\pi (p)}$. 

Now, since already on the level of $\omega$, we have $\omega_p(X^i_p,X^j_p)=0$ for all $i,j\in\{0,\ldots,n\}$ and all $p\in\mathbb{C}^{n+1}\setminus\{0\}$, it is immediate that $(\omega_{FS})_{\pi(p)}(\pi_*X^i_p,\pi_*X^j_p)=0$ for all $i,j$ and for all $\pi(p)\in\mathbb{C}P^n$. Hence we have that $(T_pT^n)^\varepsilon\supseteq T_pT^n~\forall p\in T^n$. Since equality then has to hold by dimensional analysis, we have $T^n$ is a Lagrangian submanifold.
\end{proof}

Now consider the standard action of $U\in U(n+1)$ on $\mathbb{C}^{n+1}$. Note that $U$ leaves $\omega$ invariant, since $\sum_i d(Uz)^i\wedge d\overline{Uz}^i=\sum_{ijk} U_{ij}\overline{U}_{ik} dz^j\wedge d\overline{z}^k=\sum_i dz^i\wedge d\overline{z}^i$. Furthermore, since $U$ leaves the norm invariant by definition, we have that $U^*\omega_{FS}=\omega_{FS}$, where $U^*$ is the pullback associated with the map $U$. This means that any unitary $U\in U(n+1)$ corresponds to a symplectomorphism of $\mathbb{C}P^{n}$. Since it is well-known that the complex projective space is simply connected and closed, its Hamiltonian symplectomorphism correspond to its symplectomorphism. Hence:
\begin{prop} \label{prop:unitham}
We have $U(n+1)\subset \mathrm{Ham}(\mathbb{C}P^n,\omega_{FS})$, where the identification is achieved by considering the standard action of $U$ on $\mathbb{C}^{n+1}$.
\end{prop}
\end{document}